\documentclass[letterpaper, 10 pt, conference]{ieeeconf}  

\IEEEoverridecommandlockouts                              

\usepackage{pbalance}
\usepackage{amsmath}
\usepackage{amssymb}
\usepackage{color}
\usepackage{graphicx}
\usepackage{lipsum}

\usepackage{amsthm}

\newcommand{\R}{\mathbb{R}}
\renewcommand{\t}[1]{\tilde{#1}}
\renewcommand{\c}[1]{\mathcal{#1}}
\newcommand{\h}[1]{\hat{#1}}

\newtheorem{lemma}{Lemma}
\newtheorem{theorem}{Theorem}

\newtheorem{corollary}{Corollary}
\theoremstyle{definition}
\newtheorem{example}{Example}
\newtheorem{assumption}{Assumption}
\newtheorem{definition}{Definition}
\newtheorem{remark}{Remark}

\newtheorem*{examplecontinner}{Example \currentref\ Continued}
\newenvironment{examplecont}[1]
  {\def\currentref{\ref{#1}}\begin{examplecontinner}}
  {\end{examplecontinner}}

\title{\LARGE \bf
A Recursive CBF Framework for Safety under State Uncertainty
}

\author{Rahal Nanayakkara$^{1}$, Aaron D. Ames$^{2}$ and Paulo Tabuada$^{1}$
\thanks{*This work is supported by TII under project \#A6847.}
\thanks{$^{1}$R. Nanayakkara and P. Tabuada are with the Electrical and Computer Engineering Department, University of California at Los Angeles, USA,
        {\tt\small rahaln@ucla.edu, tabuada@ee.ucla.edu}}%
\thanks{$^{2}$A. D. Ames is with the Department of Mechanical and Civil Engineering, California Institute of Technology, USA,
        {\tt\small ames@caltech.edu}}%
}

\begin{document}

\maketitle
\thispagestyle{empty}
\pagestyle{empty}

\begin{abstract}

The practical implementation of Control Barrier Functions (CBFs) for safety-critical control is often hindered by uncertainty in the knowledge of the state. While existing robust CBF methods address state uncertainty, they often lack recursive feasibility guarantees or fail when uncertainty levels are high, allowing the system to enter regions where no safe control input exists. To resolve this, we propose a novel framework of enforcing \emph{recursive CBFs}. Rather than merely ensuring the invariance of the original safe set, this approach enforces the forward invariance of a subset of the safe region where a robustly safe control input is guaranteed to exist. This holistic framework ensures that the system never strays into ambiguous regions, providing continued feasibility and safety guarantees, regardless of the level of state uncertainty.
\end{abstract}

\section{INTRODUCTION}

The transition from theoretical safety guarantees to robust real world performance remains a significant challenge in modern control system design.
While the framework of Control Barrier Functions (CBFs) \cite{cbf_main, cbf_journal} provides an elegant formalisation of safety by continually enforcing a state-dependent constraint on the control input, its practical implementation may be hindered by various uncertainties.
In particular, since most real world systems use noisy measurements to estimate the state and compute control inputs, ensuring that CBF guarantees still hold under state estimation errors remains a problem of paramount importance.

One approach for providing stronger guarantees when uncertainties are present is to strengthen the CBF inequality.
The pioneering works in this sense are \cite{jankovic2018robust}, which provides guarantees for the original safe set under small disturbances,
and \cite{issf, tissf} which provide guarantees for inflated safe sets, where inflation is proportional to the disturbance.
The recent work \cite{nanayakkara2025safety} proposed a unifying framework to combine and generalize these results, and also showed how the resultant modification provides robustness to state uncertainty under compactness assumptions on the safe set.
While works such as \cite{cosner2023robust, ramadan2024control} offer stochastic safety guarantees under probabilistic state uncertainty, their applicability is limited in scenarios that demand absolute safety.

An alternative to strengthening the CBF inequality is to
enforce the CBF inequality over all possible states.
This has been enabled by the development of observers that provide not only an estimate of the state but also a bounded set that is guaranteed to contain the true state \cite{alamo2005guaranteed, jaulin2002nonlinear, silvestre2024nonlinear}.
The framework in \cite{mrcbf2021guaranteeing, mrcbf2_iros} provided the first result of this nature by introducing Measurement Robust CBFs (MR-CBFs), where the state-dependent coefficients of the CBF inequality are bounded using their Lipschitz constants.
This framework is expanded using interval analysis with higher-order Taylor models in \cite{zhang2022control} and to handle uncertainty in the system dynamics in \cite{lindemann2024learning}.
The reference \cite{agrawal2022safe} also follows a similar approach,
but the proposed algorithm only works when the state-dependent control gain appearing in the CBF inequality is sign-definite and bounded away from zero.
However, over-approximating uncertainty using Lipschitz constants can be overly conservative, resulting in infeasibility of the constraint even when a valid control input exists.
The recent work \cite{tan_duality_based} remedies this by using a duality based approach to compute a safe control input over the exact set of uncertain states with minimal conservatism.

While the aforementioned methods propose various algorithms to compute safe control inputs, a crucial gap still remains in enforcing safety online.
These existing results either require global prior knowledge of state uncertainty across the entire safe set \cite{mrcbf2021guaranteeing, mrcbf2_iros}, or they lack recursive feasibility guarantees altogether \cite{tan_duality_based}.
This is due to the fact that, unless restrictive assumptions are imposed on the system dynamics, as in \cite{agrawal2022safe}, there may always be a level of state uncertainty that causes the CBF inequality to be infeasible in certain regions of the safe set.

In this paper, we resolve this issue through a more holistic framework. 
Rather than selecting a control input that only guarantees invariance of the original safe set, we also enforce forward invariance of the region where a safe input is guaranteed to exist despite state estimation errors. 
This is particularly important for ensuring continued feasibility of algorithms such as \cite{tan_duality_based}, or when the uncertainty exceeds the maximum allowable limit in \cite{mrcbf2021guaranteeing}. 
To this end, we propose a framework of recursive CBFs that prevents the system from entering regions where any CBF inequality becomes infeasible. 
Structurally, our methodology shares similarities with High Order CBFs (HOCBFs) \cite{nguyen2016exponential, xu2018high, xiao2019high2, xiao2021high}, leveraging a nested invariance principle analogous to that used for systems with high relative degree constraints.

\section{PRELIMINARIES}

\subsection{Notation}
A continuous function $\alpha : \R_{\geq0} \to \R_{\geq0}$ is said to be class $\c K$ if $\alpha(0)=0$ and it is strictly monotonically increasing. 
Given a function $h:\R^n \to \R$, we denote by $\nabla_x h(x)$ its gradient as a co-vector, and by $\nabla_x^2h(x)$ its Hessian. 
Given vectorfields $f,g : \R^n \to \R^n$, we use the following notation for Lie derivatives; $L_fh(x) = \nabla_xh(x) f(x)$, $L_fL_gh(x) = L_f(L_gh(x))$ and $L_g^2h(x) = L_gL_gh(x)$.
Higher order Lie derivatives are defined accordingly, and we adopt the convention that $L_f^0h(x)=h(x)$.
For a symmetric matrix $Q \in \R^{n \times n}$, $\lambda_{\min}(Q)$ denotes the smallest eigenvalue.
Proofs of all lemmas and theorems, together with an additional example, are provided in the Appendix.

Consider a single input system of the form:
\vspace{-3pt} \begin{equation}
    \dot x = f(x) + g(x) u,
    \label{eq:si_sys} \vspace{-3pt}
\end{equation}
where $x \in \R^n$, $u \in \R$ and $f,g : \R^n \to \R^n$ are locally Lipschitz continuous. 
Given an initial condition $x_0$ at time $t=0$, and a control signal $u : [0, \infty) \to \R$, we denote by $x(t, x_0, u)$ the solution to \eqref{eq:si_sys} at any time $t \in \R_{\geq0}$.
We formalize safety as forward invariance in the following definition.
\begin{definition}
    Given a locally Lipschitz continuous controller $k : \R^n \to \R$, the set $\c S \subseteq \R^n$ is said to be forward invariant (safe) under the closed-loop dynamics defined by \eqref{eq:si_sys} with $u=k(x)$ if\footnote{We assume the solutions to \eqref{eq:si_sys} with $u=k(x)$ are forward complete.}:
    \vspace{-5pt} \begin{equation}
        x_0 \in \c S \implies x(t, x_0) \in \c S, \quad \forall t \geq 0.
        \label{eq:fwd_invariance}
    \end{equation}
\end{definition}

\subsection{CBF Background}

We approach safety from the perspective of Control Barrier Functions, where a safe set $\c S$ is defined as the 0-superlevel set of a continuously differentiable function \mbox{$h:\R^n \to \R$}, i.e.:
\vspace{-5pt} \begin{equation}
    \c S = \{ x \in \R^n : h(x) \geq 0 \}.
    \label{eq:safe_set} \vspace{-3pt}
\end{equation}

In the remainder of this manuscript we assume that $\c S$ is connected. This assumption carries no loss of generality because, in cases where $\mathcal{S}$ is disconnected, forward invariance must hold for each connected component. Consequently, our proofs can be extended to the disconnected case by treating each component as an isolated safe set.
\begin{definition} (Control Barrier Function \cite{cbf_main})
    A continuously differentiable function $h : \R^n \to \R$ is a CBF for \eqref{eq:si_sys} if there exists 
    a class $\c K$ function $\alpha$ such that\footnote{Note that the strict inequality ensures that $\nabla_xh(x) \neq 0$ when $h(x)=0$.}:
    \vspace{-2pt}
    \begin{equation}
        \sup_{u \in \R} L_fh(x) + L_gh(x)u +\alpha(h(x)) > 0, \quad \forall x \in \c S.
        \label{eq:cbf_def}
    \end{equation}
\end{definition}

Since we analyze the case of unconstrained input, the condition \eqref{eq:cbf_def} is equivalent to:
\vspace{-2pt}
\begin{equation}
    L_gh(x) = 0 \implies L_fh(x) + \alpha(h(x)) > 0.
    \label{eq:cbf_def_2}
\end{equation}
An input $u$ is said to satisfy a CBF condition at some $x$ if:
\vspace{-2pt}
\begin{equation}
    \underbrace{L_fh(x) +\alpha(h(x))}_{a(x)} + \underbrace{L_gh(x)}_{b(x)}u \geq 0.
    \label{eq:cbf_condition}
\end{equation}
The set of pointwise feasible inputs can now be defined as:
\vspace{-2pt} 
\begin{equation}
    K_{\text{safe}} (x) = \{u \in \R : a(x)+b(x)u \geq 0\}.
    \label{eq:K_safe}
\end{equation}

\begin{theorem}
    If $h:\R^n \to \R$ is a CBF for \eqref{eq:si_sys}, then any locally Lipschitz controller $k:\R^n \to \R$ that satisfies $k(x) \in K_\text{safe}(x)$ for all $x \in \c S$ renders the set $\c S$ safe for the closed-loop dynamics given by \eqref{eq:si_sys} with $u=k(x)$ \cite{cbf_main}.
\end{theorem}

The preceding theorem essentially states that, in order to construct a control law that renders the set $\c S$ safe, we must always be able to select an input $u \in \R$ satisfying \eqref{eq:cbf_condition}.

\subsection{Measurement Uncertainty}

The problem of enforcing safety becomes significantly more challenging when the true state of the system is unknown. 
For instance, suppose we only have an estimate $\h x$ of the true state $x$, based on noisy/partial observations of the system's trajectory, and the guarantee that the state lies within some known compact neighbourhood of the estimate. 
We formalize this notion using the following definitions.

\begin{definition} (State Uncertainty Map)
    A set-valued map $B : \R^n \times \R_{\geq 0} \to 2^{\R^n}$, is called a state uncertainty map if $B(\h x, t)$ is compact and $\h x \in B(\h x, t)$, for all \mbox{$(\h x, t) \in \R^n \times \R_{\geq 0}$}.
\end{definition}

Given a state estimate $\h x$, the set-valued function $B(\h x, t)$ defines the uncertainty region around $\h x$, encompassing all possible values the true state may attain at time $t$.

\begin{definition}
    We say $\h x$ is a consistent estimate of the state of \eqref{eq:si_sys} at time $t$ if $ x(t) \in B(\h x, t)$,
    where $x(t)$ is the true state.
\end{definition}

Observers that provide deterministic error bounds, such as \cite{alamo2005guaranteed, jaulin2002nonlinear, silvestre2024nonlinear}, provide both a consistent estimate $\h x$ of the state, and a bounded set $B(\h x, t)$ guaranteed to contain the true state $x$.
When no additional information is known, a sufficient condition to enforce safety at some $(\h x, t) \in \R^n \times \R_{\geq 0}$ is to force the control input $u$ to satisfy:
\vspace{-2pt} 
\begin{equation}
    a(x)+b(x)u \geq 0 \quad \forall x \in B(\h x, t) \cap \c S.
    \label{eq:cbf_inequality_b}
    \vspace{-2pt} 
\end{equation}

\begin{remark}
    We note that \eqref{eq:cbf_inequality_b} need only hold for $x \in B(\hat{x}, t) \cap \mathcal{S}$ rather than the entire ball $B(\hat{x}, t)$. 
    While the state estimator may provide an uncertainty bound containing states outside the safe set, a robustly safe controller ensures that the true state never leaves $\mathcal{S}$. 
    Consequently, the closed-loop system allows us to ``shrink" the effective uncertainty set from $B(\hat{x}, t)$ to $\t B(\hat{x}, t) = B(\hat{x}, t) \cap \mathcal{S}$.
\end{remark}

Accordingly, we define the pointwise set of robustly safe inputs based on $(\h x, t)$ as:
\vspace{-2pt} 
\begin{equation}
    \h K(\h x, t) = \{ u \in \R :a(x)+b(x)u \geq 0, \, \forall x \in \t B(\h x, t) \}.
    \vspace{-2pt}
    \label{eq:robustly_safe_inputs}
\end{equation}

\begin{theorem} \label{thm:robustly_safe}
    If $h:\R^n \to \R$ is a CBF and $B : \R^n \times \R_{\geq 0} \to 2^{\R^n}$ is a state uncertainty map, then any locally Lipschitz controller\footnote{We assume that $t \mapsto k(\h x(t), t)$ is locally Lipschitz continuous to guarantee existence and uniqueness of solutions.}
    $k : \R^n \times \R \to \R$ that satisfies $k(\h x, t) \in \h K(\h x, t)$ for all $\h x \in \R^n$ and $t \in \R_{\geq 0}$, renders the set $\c S$ safe for the closed-loop dynamics given by \eqref{eq:si_sys} with $u=k(\h x, t)$, for any consistent estimate $\h x$.
\end{theorem}

The set $\h K(\h x, t)$ varies depending on the ``level of uncertainty''.
We illustrate this in the following example.

\begin{figure*}[t]
    \centering
    \begin{minipage}{0.32\textwidth}
        \centering
        \includegraphics[width=\linewidth, height=0.67\linewidth]{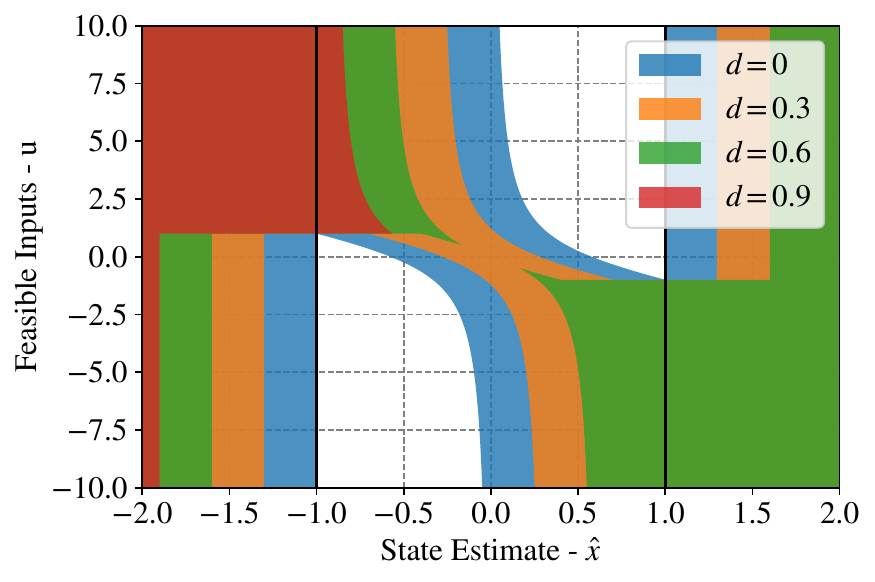}
        \caption{Set of feasible inputs $\h K(\h x)$ for varying levels of state uncertainty.}
        \label{fig:scalar_feasible_inputs}
    \end{minipage}
    \hfill
    \begin{minipage}{0.32\textwidth}
        \centering
        \includegraphics[width=\linewidth, height=0.67\linewidth]{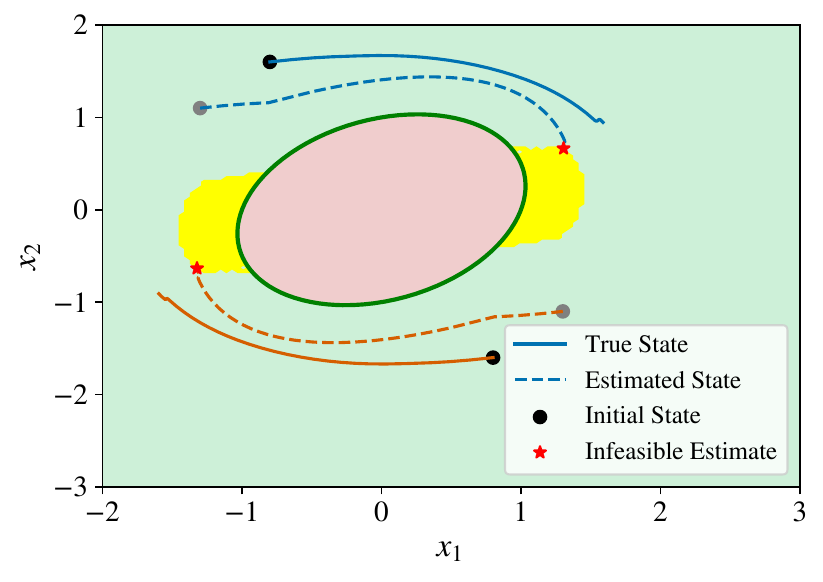}
        \caption{Trajectories end up entering the region where no feasible input exists, shown in yellow. 
        }
        \label{fig:infeasible_convergence}
    \end{minipage}
    \hfill
    \begin{minipage}{0.32\textwidth}
        \centering
        \includegraphics[width=\linewidth, height=0.67\linewidth]{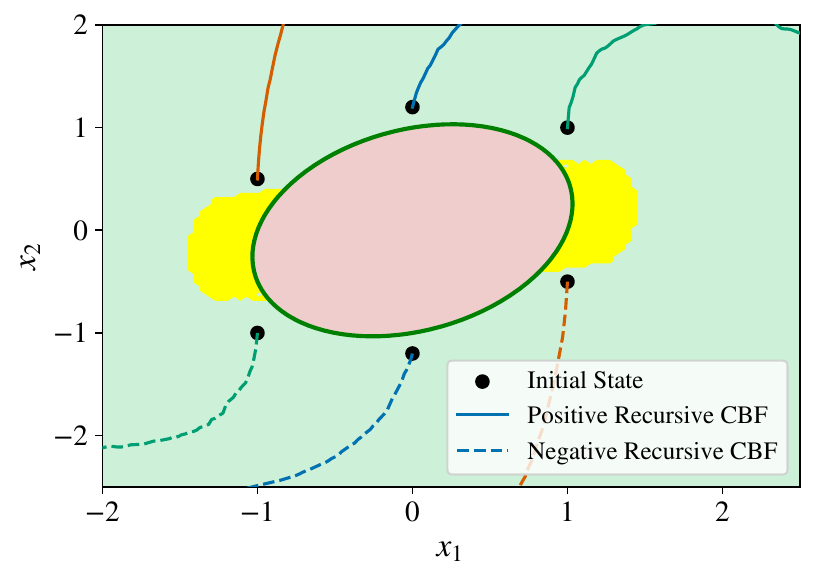}
        \caption{Closed-loop trajectories of \eqref{eq:example_lin_sys}, for both positive and negative Recursive CBFs.
        }
    \label{fig:2d_sys_example}
    \end{minipage}
\end{figure*}

\begin{example}
Consider the scalar system $\dot x = x + u$ and the CBF $h(x)=1-x^2$ for the class $\c K$ function \mbox{$\alpha(r)=r$}. We consider the set of feasible inputs \eqref{eq:robustly_safe_inputs} for the case when $B(\h x, t) = [\h x-d, \h x+d]$ where $d \in \R_{\geq 0}$.
Fig \ref{fig:scalar_feasible_inputs} depicts the plot of the set-valued function $\h K(\h x)$ for varying levels of $d$.
The safe set $\c S = [-1, 1]$ is the region between the black lines. Note how $\h K(\h x) = \R$ when $B(\h x) \cap \c S = \emptyset$, since the condition in the definition of \eqref{eq:robustly_safe_inputs} becomes vacuously true. It is also possible to observe how $\h K(\h x) = \emptyset$ for larger values of $d$, when $\h x$ lies close to the origin.

\end{example}

A sufficient condition for the set $\h K(\h x, t)$ to always be non-empty is given in the following theorem.

\begin{theorem} \label{thm:lgh_no_sign_change}
    If either $L_gh(x) \geq 0$ for all $x \in \c S$ or $L_gh(x) \leq 0$ for all $x \in \c S$, then $\h K(\h x, t)$ is non-empty for all \mbox{$(\h x, t) \in \R^n \times \R_{\geq 0}$}.
\end{theorem}

\begin{corollary} \label{cor:lgh_eps}
    If there exists $\varepsilon \in \R_{>0}$ and $s \in \{-1, 1\}$ such that $s L_gh(x) \geq \varepsilon$ for all $x \in \c S$, then:
    \vspace{-5pt} 
    \begin{equation}
        \left\{u \in \R : s u \geq - \frac{1}{\varepsilon}\min \left\{ 0, \inf_{x \in \t B(\h x, t)} a(x)\right\} \right\} \subseteq \h K(\h x, t),
        \vspace{-5pt} 
    \end{equation}
    for all $(\h x, t) \in \R^n \times \R_{\geq 0}$.
\end{corollary}

In fact, the work in \cite{agrawal2022safe} provides a similar construction for the multi-input case, under the assumption that each element of $L_gh(x)$ is sign-definite and always bounded away from zero.
We will simply use this observation as a stepping stone for our main results in Section \ref{sct:main_results}.

While existing works such as \cite{tan_duality_based} have demonstrated how a safe input can be computed when the set $\h K(\h x, t)$ is non-empty, the harder problem at hand has not been answered: how do we ensure that this set will always remain non-empty for all time $t\geq0$ along trajectories.

Given a system, a CBF, and a level of uncertainty, there may always exist regions in the state space where $\h K(\h x, t)$ is empty. While this region may be ``small'' in comparison to the entire safe set, it could be the case that a robustly safe control law always ends up bringing trajectories into this region. We illustrate this in the following example.

\begin{example} \label{ex:2d_example}

Consider the linear system:
\vspace{-2pt} 
\begin{equation}
\begin{aligned}
    \dot x_1 &= x_2 \\
    \dot x_2 &=-x_1 -x_2+u,
\end{aligned} \label{eq:example_lin_sys}
\vspace{-2pt} 
\end{equation} 
and the function $h(x) = x_1^2+x_2^2-\frac{1}{2}x_1 x_2-1$, which is a CBF for the class $\c K$ function $\alpha(z)=z$. This yields $a(x)=\frac{3}{2}(x_1^2-x_2^2)-1$ and $b(x)=2x_2-\frac{1}{2}x_1$. We consider the case where $B(\h x, t) = \{x \in \R^2 : \|x-\h x \|_\infty \leq 0.5\}$ for all $t \in \R_{\geq0}$, i.e., constant state uncertainty given by an infinity norm ball. The yellow shaded region in Fig \ref{fig:infeasible_convergence} represents the values of $\h x$ for which the set \eqref{eq:robustly_safe_inputs} is empty, meaning no feasible control input exists that renders the system safe.
We simulate the system under the minimum-norm controller, computed using the method in \cite{tan_duality_based}, to ensure that \eqref{eq:robustly_safe_inputs} is always satisfied when non-empty. 
However, we see in Fig \ref{fig:infeasible_convergence} that trajectories starting from multiple initial conditions end up in the yellow region, where a safe control input no longer exists, and the controller from \cite{tan_duality_based} is ill-defined. 
This is due to the fact that while methods such as \cite{tan_duality_based} are effective at constructing a controller when it exists, they provide no recursive feasibility guarantees.
Furthermore, alternative approaches like MR-CBFs \cite{mrcbf2021guaranteeing} fail entirely here, as a valid MR-CBF cannot even be constructed for this level of uncertainty.

\end{example}

This phenomenon observed in Example \ref{ex:2d_example} is not merely an artifact of the method used to compute the control input, 
but rather a fundamental consequence of the fact that sets which are controlled invariant are not necessarily robustly controlled invariant under state uncertainty.
Thus, in order to guarantee continued existence of a safe control input, the controller must not only render the set $\c S$ safe but also keep the system within the region where $\h K(\h x, t)$ is non-empty.
In the remainder of this paper, we call such a subset of $\c S$ a ``robustly safe'' set, as formalized in the following definition.

\begin{definition}
    We say a set $\c S$ can be made robustly safe under state uncertainty if for every $x \in \c S$, 
    the set \eqref{eq:robustly_safe_inputs} is non-empty for every consistent estimate $\h x$ and every $t \in \R_{\geq 0}$.
\end{definition}

This definition essentially states that if a set $\c S$ can be made robustly safe, then there always exists a safe controller $k(\h x, t)$ for any level of state uncertainty. 
In conjunction with Theorem \ref{thm:robustly_safe}, this establishes that the set $\c S$ can be rendered forward invariant subject to any compact uncertainty bound $B(\h x, t)$.
However, as illustrated in Example \ref{ex:2d_example}, it is generally impossible to make the entire original safe set $\c S$ robustly safe. Thus, in the subsequent sections, we develop a systematic framework to construct robustly safe subsets of $\c S$.

\section{RECURSIVE CONTROL BARRIER FUNCTIONS} \label{sct:main_results}

Based on the discussion above, the natural question one may ask is if a second barrier can be employed to restrict the trajectories of the system to the region where $\h K(\h x, t)$ is non-empty. 
While indeed this is possible, one must now ensure that the robustly feasible set of inputs for this new barrier is also non-empty. 
For instance, if we enforce feasibility of the original barrier by requiring $L_gh(x) > 0$, using some barrier $h_2 : \R^n \to \R$, we also need to ensure that $h_2$ is robustly feasible.
We may again choose to do this using a third barrier that forces $L_gh_2(x) > 0$ and the argument quickly becomes \emph{recursive}.

Indeed there do exist cases when this recursive reasoning bears fruit. 
We first discuss a simple example where such a construction can be performed before presenting a more general result in the next section.

\begin{examplecont}{ex:2d_example}
To illustrate this idea, consider again the system and CBF discussed in Example \ref{ex:2d_example}, and recall that $L_gh(x) = 2x_2-\frac{1}{2}x_1$. Motivated by Theorem \ref{thm:lgh_no_sign_change}, we define a second barrier $h_2(x) = L_gh(x) - \varepsilon$, where $\varepsilon \in \R_{>0}$ is some small positive constant. 
We first observe that $h_2$ is a CBF for any choice of class $\c K$ function, since $L_gh_2 = L_g^2h = 2 > 0$ ensuring that the antecedent in \eqref{eq:cbf_def_2} is always false. Now, by Corollary \ref{cor:lgh_eps}, any $u \in \R$ satisfying:
\vspace{-5pt} 
\begin{equation*}
    u \geq -\frac{1}{2} \min
    \left\{0, \inf_{x \in \t B(\h x, t)} L_fh_2(x)+\alpha_2(h_2(x)) \right\} =: \overline{u}_2,
    \vspace{-3pt} 
\end{equation*}
where $\alpha_2$ is the class $\c K$ function corresponding to $h_2$, renders the set $\{x \in \R^2 : L_gh(x) \geq \varepsilon\}$ invariant. Within this invariant set, again by Corollary \ref{cor:lgh_eps}, we are guaranteed that any $u$ satisfying:
\vspace{-5pt} 
\begin{equation*}
    u \geq -\frac{1}{\varepsilon} \min 
    \left\{0, \inf_{x \in \t B(\h x, t)} \frac{3}{2}(x_1^2-x_2^2)-1 \right\} =: \overline{u}_1,
    \vspace{-3pt} 
\end{equation*}
ensures the satisfaction of the original CBF $h$. Thus by choosing any $u \geq \max\{  \overline{u}_1,  \overline{u}_2\}$, we can guarantee satisfaction of both CBFs, $h$ and $h_2$, and hence the forward invariance of $\c S \cap \{x \in \R^2 : L_gh(x) \geq \varepsilon\}$.
\end{examplecont}

\subsection{General Results}

We first state the assumption under which our results are derived.

\begin{assumption}
\label{assump:nonzero}
    There exists $\gamma \in \mathbb Z_{>0}$ and $\varepsilon_\gamma \in \R_{>0}$
    such that $|L_g^ \gamma h(x)| \geq \varepsilon_\gamma$ for all $x \in \c S$.
\end{assumption}
 Later, in Section \ref{sct:polynomial} and \ref{sct:convex}, we give examples of classes of systems that satisfy this assumption.
Since $L_g^ \gamma h : \R^n \to \R$ is a continuous function of $x$ and $\c S$ is connected, we note that $|L_g^ \gamma h(x)| \geq \varepsilon_\gamma$ necessarily implies that either $L_g^ \gamma h(x) \geq \varepsilon_\gamma$ for all $x \in \c S$ or $L_g^\gamma h(x) \leq -\varepsilon_\gamma$ for all $x \in \c S$.
With this assumption in mind, we now propose two possible constructions that can be performed depending on the parity of $\gamma$ and the sign of $L_g^ \gamma h(x)$.

\begin{definition}
\label{def:pos_recursiveCBF}
Under Assumption \ref{assump:nonzero}, a \textbf{positive recursive CBF} is a series of functions $h_1, \dots, h_\gamma : \R^n \to \R$ satisfying:
\vspace{-5pt} 
\begin{equation}
    \begin{aligned}
        h_1(x) &= h(x) \\
        h_k(x) &= L_gh_{k-1}(x)-\varepsilon_{k-1}, \quad  k=2,\dots,\gamma,
      \end{aligned}
      \label{eq:lgh_pos_cons}
      \vspace{-3pt} 
    \end{equation}
for any $\varepsilon_1, \dots, \varepsilon_{\gamma-1} \in \R_{>0}$.
The set $\c S_k$ is defined as \mbox{$\c S_k = \{ x\in\mathbb{R}^n : h_k(x)\geq0\}$} for each $k=1,\dots,\gamma$ and $\c C_\gamma$ is defined by:
\vspace{-8pt} \begin{equation}
    \c C_\gamma := \bigcap_{k=1}^{\gamma}\c S_k.
    \label{eq:intersect_set}
    \vspace{-2pt} 
\end{equation}
\end{definition}

\begin{assumption}
    The set $\c C_\gamma$ in \eqref{eq:intersect_set} is non-empty.
\end{assumption}
\begin{remark} \label{rem:C_non_empty}
    The set $\c C_k$, for any $k \geq 2$, is non-empty if there exists $x_1, x_2 \in \c C_{k-1}$ such that $L_g^{k-1}h(x_1) \geq \varepsilon_{k-1}$ and $L_g^{k-1}h(x_2) \leq -\varepsilon_{k-1}$. 
    This essentially means that $L_g^{k-1}h(x)$ changes sign within $\c C_{k-1}$, since $\varepsilon_{k-1}$ can be chosen to be arbitrarily small.
    Consequently, an alternative to Assumption \ref{assump:nonzero} is to constructively define $\gamma$ as the smallest positive integer for which $L_g^\gamma h(x)$ is sign-definite over $\c C_\gamma$ (instead of all of $\c S$) guaranteeing the non-emptiness of $\c C_\gamma$. Under this condition, all subsequent theoretical guarantees remain identical.
\end{remark}

\begin{lemma}
    For each $k=1, \dots, \gamma-1$, it holds that $L_gh_k(x) \geq \varepsilon_k$ for all $x \in \c C_\gamma$.
    \label{lem:pos_cons}
\end{lemma}

Lemma \ref{lem:pos_cons} essentially states that each $h_k$ for \mbox{$k=1,\dots,\gamma-1$}
is a CBF for the system \eqref{eq:si_sys} within $\c C_\gamma$,
since $L_gh_k(x) > 0$ implies that the antecedent in \eqref{eq:cbf_def_2} is always false.
We also note that the Lie derivatives of the CBFs $h_1, \dots, h_k$ relate to that of the original CBF $h$ as: 
\begin{equation}
\begin{aligned}
    L_gh_k(x) &= L_g^kh(x) \\
    L_fh_k(x) &= L_fL_g^{k-1}h(x).
\end{aligned} \label{eq:higher_lie_pos}
\end{equation}
We later show that under appropriate assumptions on $\gamma$ and $L_g^\gamma h$, enforcing the barriers $h_1,\dots,h_\gamma$ guarantees the existence of a robustly safe input and the robust forward invariance of $\c C_\gamma$ under state uncertainty.

\begin{definition}
\label{def:neg_recursiveCBF}
Under Assumption \ref{assump:nonzero}, a \textbf{negative recursive CBF} is a series of functions $\t h_1, \dots, \t h_\gamma : \R^n \to \R$ satisfying:
\begin{equation}
    \begin{aligned}
        \t h_1(x) &= h(x) \\
        \t h_k(x) &= -L_g \t h_{k-1}(x)-\varepsilon_{k-1}, \quad k=2,\dots,\gamma,
        \label{eq:lgh_neg_cons}
      \end{aligned}
    \end{equation}
for any $\varepsilon_1, \dots, \varepsilon_{\gamma-1} \in \R_{>0}$.
The set $\t {\c S}_k$ is defined as \mbox{$\t {\c S}_k = \{ x\in\mathbb{R}^n : \t h_k(x)\geq0\}$}
for each $k=1,\dots,\gamma$ and $\t {\c C}_\gamma$ is defined by:
\vspace{-8pt} \begin{equation}
    \t {\c C}_\gamma = \bigcap_{k=1}^{\gamma} \t{\c S}_k.
    \label{eq:intersect_set_2}
\end{equation}
\end{definition}

\begin{assumption}
    The set $\t {\c C}_\gamma$ in \eqref{eq:intersect_set_2} is non-empty.\footnote{Note that a version of Remark \ref{rem:C_non_empty} also applies to $\t {\c C}_\gamma$.}
\end{assumption}

\begin{lemma}
    For each $k=1,\dots,\gamma-1$, it holds that 
    $L_g\t h_k(x) \leq -\varepsilon_k$ for all $x \in \t{\c C}_\gamma$.
    \label{lem:neg_cons}
\end{lemma}

Similar to the previous case, $\t h_k$ for $k=1,\dots,\gamma-1$ is a CBF within $\t{\c C}_\gamma$. The corresponding Lie derivatives are:
\begin{equation}
\begin{aligned}
    L_g \t h_k(x) &= (-1)^{k-1} L_g^k h(x) \\
    L_f \t h_k(x) &= (-1)^{k-1} L_f L_g^{k-1} h(x).
\end{aligned}   \label{eq:higher_lie_neg}
\end{equation}

\begin{remark}
The recursive CBFs introduced in Definitions \ref{def:pos_recursiveCBF} and \ref{def:neg_recursiveCBF} bear a structural resemblance to High Order CBFs (HOCBFs), as both approaches construct a sequence of barrier functions whose intersected 0-superlevel sets must be rendered forward invariant. 
However, their underlying motivations and mathematical formulations differ significantly. 
HOCBFs were specifically developed to address safety constraints with high relative degree, where the control input does not explicitly appear until the $\gamma$-th derivative of the barrier function, i.e., there exists some $\gamma \in \mathbb Z_{>0}$ such that $L_g L_f^{k-1} h = 0$ for $k=1,\dots,\gamma-1$. 
The early formulations \cite{nguyen2016exponential, xu2018high} of HOCBFs proposed successive functions of the form \mbox{$h_k = L_f h_{k-1} + a_k h_{k-1}$}, where $a_k \in \R_{>0}$, often referred to as exponential HOCBFs.
This was later generalized in \cite{xiao2019high2, xiao2021high} where the constants $a_k$ were replaced by class $\c K$ functions $\alpha_k$.
Notably, in the HOCBF framework, the control input only appears in the derivative of the final function $h_\gamma$, and enforcing this single condition guarantees the forward invariance of the intersection of all 0-superlevel sets. 
In contrast, recursive CBFs are designed to overcome the sign indefiniteness of the control gain $L_g h$ under state uncertainty. 
Consequently, each successive barrier is defined in terms of $L_g h_{k-1}$, rather than $L_f h_{k-1}$, and now the control input appears in every resulting CBF inequality. 
As shown in the subsequent analysis, satisfying all of these inequalities simultaneously guarantees robust forward invariance.
\end{remark}

We now utilize both types of recursive CBFs to demonstrate how a subset of $\c S$ can be made forward invariant while also ensuring that a robustly safe input always exists. 
\vspace{-2pt}
\begin{theorem}
    If Assumption \ref{assump:nonzero} holds with odd $\gamma$, then either $\c C_\gamma$ or $\t {\c C}_\gamma$ can be made robustly safe under state uncertainty. Specifically:
    \begin{enumerate}
        \item if $L_g^\gamma h > 0$ then $\c C_\gamma$ is robustly safe.
        \item if $L_g^\gamma h < 0$ then $\t {\c C}_\gamma$ is robustly safe.
    \end{enumerate}
    \label{thm:odd_gamma}
    \vspace{-3pt}
\end{theorem}

The preceding Theorem states that the sign of $L_g^\gamma h$ determines which set $\c C_\gamma$ or $\t{\c C}_\gamma$ can be made robustly safe. 
While $h_1,\dots, h_\gamma$ and $\t h_1,\dots,\t h_\gamma$ will all be CBFs on both $\c C_\gamma$ and $\t{\c C}_\gamma$ for both cases $L_g^\gamma h>0$ and $L_g^\gamma h<0$, the set of feasible inputs satisfying all CBF inequalities may be empty if the appropriate type of recursive CBF is not used for the appropriate sign of $L_g^\gamma h$, as illustrated below.

\begin{remark} \label{rem:wrong_cons}
For instance, consider the case when $\gamma$ is odd, $L_g^\gamma h \leq -\varepsilon_\gamma$ and we attempt to make $\c C_\gamma$ robustly safe by using a positive recursive CBF \eqref{eq:lgh_pos_cons}.
By Lemma \ref{lem:pos_cons} and Corollary \ref{cor:lgh_eps}, any input 
$u \geq \underline{u}_k$
satisfies the CBF inequality corresponding to $h_k$ for $k=1,\dots, \gamma-1$.
However, since $L_g^\gamma h \leq -\varepsilon_\gamma$, by Corollary \ref{cor:lgh_eps}, the $\gamma^\text{th}$ CBF inequality is satisfied by an input $u \leq \overline{u}_\gamma.$ Thus to render $\c C_\gamma$ robustly safe, $u$ must satisfy:
\vspace{-5pt}\begin{equation*}
\max_{k=1,\dots,\gamma-1} \{ \underline{u}_k \} \leq u \leq \overline{u}_\gamma, \vspace{-3pt}
\end{equation*}
which, depending on the values of $\underline{u}_k$ and $\overline{u}_\gamma$ may be infeasible.
\end{remark}

\begin{theorem} \label{thm:even_gamma}
If Assumption \ref{assump:nonzero} holds with even $\gamma$, and if \mbox{$L_g^\gamma h > 0$}, then both $\c C_\gamma$ and $\t{\c C}_\gamma$ can be made robustly safe under state uncertainty.
\end{theorem}

The preceding Theorem states that both sets $\c C_\gamma$ and $\t{\c C}_\gamma$ can be made robustly safe when $\gamma$ is even and $L_g^\gamma h > 0$. In contrast, when $\gamma$ is even and $L_g^\gamma h < 0$, despite the fact that $h_1,\dots, h_\gamma$ and $\t h_1,\dots,\t h_\gamma$ will all be CBFs on both $\c C_\gamma$ and $\t{\c C}_\gamma$, the feasible set of inputs satisfying all constraints may be empty, in a manner similar to Remark \ref{rem:wrong_cons}.

In the following Sections, we discuss specific scenarios where Assumption \ref{assump:nonzero} holds.

\subsection{Polynomials} \label{sct:polynomial}

We first consider the case where $g(x)$ is a constant function, i.e., systems of the form:
\vspace{-3pt}
\begin{equation}
    \dot x = f(x) + bu,
    \vspace{-3pt}
    \label{eq:const_sys}
\end{equation}
where $b \in \R^n$ and $\|b\| \neq 0$.

\begin{theorem} \label{thm:poly_h_const_g}
    If $h : \R^n \to \R$ is a polynomial of degree $\gamma$ and if it is a CBF for the system \eqref{eq:const_sys}, then there exists $a \in \R$ such that $L_g^\gamma h(x) = a$.
\end{theorem}

Theorem \ref{thm:poly_h_const_g} shows that a system of the form \eqref{eq:const_sys} with a polynomial CBF will always satisfy Assumption \ref{assump:nonzero}. 
This observation is particularly useful since a large class of mechanical systems can be expressed in the form \eqref{eq:const_sys}.

\begin{remark}
    It is possible that a system \eqref{eq:const_sys} with a polynomial CBF will satisfy Assumption \ref{assump:nonzero} for some $\gamma$ less than the degree of the polynomial, since certain even degree polynomials may always be non-zero. To illustrate this, consider the simple example where $b=1$ and $h(x)=x^3+x$. Although the polynomial degree of $h$ is $3$, we see that $L_gh(x) = 3 x^2+1 \geq 1$ for all $x \in \R$, and hence Assumption \ref{assump:nonzero} is satisfied with $\gamma=1$.
\end{remark}

We now extend the results of Theorem \ref{thm:poly_h_const_g} to the case where both $f,g : \R^n \to \R^n$ are polynomials in \eqref{eq:si_sys}. By utilizing the dynamic extension for the input, we construct the new system:
\vspace{-10pt}
\begin{equation} \label{eq:extended_sys}
    \dot z = \begin{pmatrix}
        \dot x \\ \dot u
    \end{pmatrix} = \underbrace{
    \begin{pmatrix}
        f(x)+g(x)u \\
        0
    \end{pmatrix}
    }_{F(z)} + \underbrace{\begin{pmatrix}
        0 \\ 1
    \end{pmatrix}}_{b}v,
\end{equation}
where $z=(x,u) \in \R^{n+1}$.
Given a polynomial CBF $h_0$ for the original system \eqref{eq:si_sys}, under the extended dynamics we get:
\vspace{-8pt}
\begin{equation*}
    L_bh_0(z) = \nabla_zh_0(z) b = 0,
    \vspace{-3pt}
\end{equation*}
since $h_0$ is only a function of the first $n$ elements of $z$ causing the last entry of $\nabla_zh_0(z)$ to be $0$. As the safety constraint is now of a higher relative degree for the extended system, we resort to constructing an exponential Higher Order CBF $h:\R^{n+1} \to \R$ as discussed in \cite{nguyen2016exponential}:
\begin{equation}
\begin{aligned}
     h(x,u) &= L_Fh_0(x,u) + ah_0(x) \\
     &= L_fh_0(x)+L_gh_0(x)u+ah_0(x),
\end{aligned}   \label{eq:exponential_cbf}
\end{equation}
where $a \in \R_{>0}$. Now, if $f, g$ and $h_0$ are polynomials with degrees $\gamma_f, \gamma_g$ and $\gamma_{h_0}$ respectively, we see that $h(z)$ has degree at most $\gamma_{h_0}+\max(\gamma_f-1, \gamma_g)$ since:
\begin{equation*}
    h(z) = \underbrace{\nabla_x h_0(x) f(x)}_{\text{degree} \leq \gamma_{h_0}+\gamma_f-1}
    + \underbrace{\nabla_x h_0(x) g(x) u}_{\text{degree} \leq \gamma_{h_0} + \gamma_g} + \underbrace{a h_0(x)}_{\text{degree} = \gamma_{h_0}}.
\end{equation*}

Theorem \ref{thm:poly_h_const_g} is now applicable to the extended system \eqref{eq:extended_sys} with the exponential HOCBF \eqref{eq:exponential_cbf}.

\subsection{Convex CBFs} \label{sct:convex}

In this section we illustrate another class of systems for which Assumption \ref{assump:nonzero} is satisfied.
We consider the case when the CBF $h: \R^n \to \R$ is twice differentiable and strongly convex, i.e., there exists some $\mu \in \R_{>0}$ such that:
\begin{equation}
    \nabla^2_x h(x) \succeq \mu I. \label{eq:strong_convex}
\end{equation}

To see why convexity plays a role, we expand the definition of the second Lie derivative to get:
\begin{align}
    L_g^2h(x) &= \nabla_x(\nabla_x h(x) g(x)) g(x) \nonumber \\
    &= \nabla_x h(x) \nabla_x g(x)g(x) + g^\top(x) \nabla^2_xh(x) g(x). \label{eq:lgh2}
\end{align}

\begin{corollary} \label{cor:strong_convexity}
    If $h : \R^n \to \R$ is strongly convex and if it is a CBF for system \eqref{eq:const_sys} with constant $g(x)$, then there exists $\varepsilon \in \R_{>0}$ such that $L_g^2h(x) \geq \varepsilon$ for all $x \in \R^n$.
\end{corollary}

Another scenario is the case when $g(x) = P \nabla h^\top (x)$, for some symmetric, non-singular $P \in \R^{n \times n}$. Note that this includes the case where $g(x)$ is a gradient flow of $h(x)$, i.e., $g(x)= k \nabla h(x)$, for $k \in \R, k\neq 0$.
Before stating the result, we first recall that if $h:\R^n \to \R$ is strongly convex, it possesses a unique global minimum $x^* \in \R^n$ \cite{boyd2004convex}. For the safe set to be non-trivial, we necessarily require that $x^* \notin \c S$.

\begin{corollary} \label{cor:strongly_convex_Lgh2}
    Let $h : \R^n \to \R$ be a strongly convex function with a unique minimizer $x^*$, and a CBF for \eqref{eq:si_sys} with $g(x) = P \nabla h^\top (x)$, where $P \in \R^{n \times n}$ is a symmetric, non-singular matrix. 
    If $x^* \notin \c S$, then there exists $\varepsilon \in \R_{>0}$ such that $L_g^2h(x) \geq \varepsilon$ for all $x \in \c S$.
\end{corollary}

In both of the previous cases, not only is Assumption \ref{assump:nonzero} satisfied but the conditions of Theorem \ref{thm:even_gamma} also hold, and hence its results are directly applicable to construct robustly controlled invariant sets.

\section{SIMULATIONS}

\begin{examplecont}{ex:2d_example}
Revisiting system \eqref{eq:example_lin_sys}, since $L_g^2h(x)=2$, we now see that it is possible to use both types of recursive CBFs \eqref{eq:lgh_pos_cons} and \eqref{eq:lgh_neg_cons} as stated in Theorem \ref{thm:even_gamma} to ensure that $\h K(\h x, t)$ is always non-empty. Fig \ref{fig:2d_sys_example} depicts the simulated closed-loop trajectories (with $\varepsilon=0.1$), where the appropriate type of recursive CBF is selected based on the value of the initial state, i.e., if $L_gh(x_0)>0$, \eqref{eq:lgh_pos_cons} is used and vice versa.
During simulation, a state estimate $\h x = x + e$, where $e$ is a random estimation error satisfying $\| e \|_\infty \leq 0.5$, is used to compute a robustly safe control input.
Fig \ref{fig:2d_sys_example} clearly indicates that under the proposed framework, the system avoids the unsafe region (shown in pink) and the infeasible region (shown in yellow).
\end{examplecont}

\section{CONCLUSION}

In this paper, we presented a recursive CBF framework for maintaining safety under state uncertainty. 
While existing CBF methods lacked recursive feasibility guarantees or failed when uncertainty became too large, the proposed approach ensures both feasibility and safety by enforcing forward invariance of a subset of the safe set where a safe control input is always guaranteed to exist. 
Finally, we illustrated the theoretical construction through simulations on practical systems.





\bibliographystyle{IEEEtran}
\bibliography{IEEEabrv,refs}

@article{nanayakkara2025safety,
  author={Nanayakkara, Rahal and Ames, Aaron D and P. Tabuada},
  journal={arXiv preprint arXiv:2508.17226}, 
  title={Safety Under State Uncertainty: Robustifying Control Barrier Functions},
  year={2025}
}

@article{tan_duality_based,
      title={A Duality-Based Optimization Formulation of Safe Control Design with State Uncertainties}, 
      author={Xiao Tan and Rahal Nanayakkara and Paulo Tabuada and Aaron D. Ames},
      journal={arXiv preprint arXiv:2603.26999},
      year={2026},
}

@inproceedings{nguyen2016exponential,
  title={Exponential control barrier functions for enforcing high relative-degree safety-critical constraints},
  author={Nguyen, Quan and Sreenath, Koushil},
  booktitle={2016 American Control Conference (ACC)},
  pages={322--328},
  year={2016},
  organization={IEEE}
}

@article{agrawal2022safe,
  title={Safe and robust observer-controller synthesis using control barrier functions},
  author={Agrawal, Devansh R and Panagou, Dimitra},
  journal={IEEE Control Systems Letters},
  volume={7},
  pages={127--132},
  year={2022},
  publisher={IEEE}
}

@book{boyd2004convex,
  title={Convex optimization},
  author={Boyd, Stephen and Vandenberghe, Lieven},
  year={2004},
  publisher={Cambridge university press}
}

@inproceedings{mrcbf2021guaranteeing,
  title={Guaranteeing safety of learned perception modules via measurement-robust control barrier functions},
  author={Dean, Sarah and Taylor, Andrew and Cosner, Ryan and Recht, Benjamin and Ames, Aaron},
  booktitle={Conference on Robot Learning},
  pages={654--670},
  year={2021},
  organization={PMLR}
}

@inproceedings{mrcbf2_iros,
  title={Measurement-robust control barrier functions: Certainty in safety with uncertainty in state},
  author={Cosner, Ryan K and Singletary, Andrew W and Taylor, Andrew J and Molnar, Tamas G and Bouman, Katherine L and Ames, Aaron D},
  booktitle={2021 IEEE/RSJ International Conference on Intelligent Robots and Systems (IROS)},
  pages={6286--6291},
  year={2021},
  organization={IEEE}
}

@article{jankovic2018robust,
  title={Robust control barrier functions for constrained stabilization of nonlinear systems},
  author={Jankovic, Mrdjan},
  journal={Automatica},
  volume={96},
  pages={359--367},
  year={2018},
  publisher={Elsevier}
}

@article{tissf,
  title={Safe controller synthesis with tunable input-to-state safe control barrier functions},
  author={Alan, Anil and Taylor, Andrew J and He, Chaozhe R and Orosz, G{\'a}bor and Ames, Aaron D},
  journal={IEEE Control Systems Letters},
  volume={6},
  pages={908--913},
  year={2021},
  publisher={IEEE}
}

@article{issf,
  title={Input-to-state safety with control barrier functions},
  author={Kolathaya, Shishir and Ames, Aaron D},
  journal={IEEE control systems letters},
  volume={3},
  number={1},
  pages={108--113},
  year={2018},
  publisher={IEEE}
}

@article{cbf_main,
  title={Control barrier function based quadratic programs for safety critical systems},
  author={Ames, Aaron D and Xu, Xiangru and Grizzle, Jessy W and Tabuada, Paulo},
  journal={IEEE Transactions on Automatic Control},
  volume={62},
  number={8},
  pages={3861--3876},
  year={2016},
  publisher={IEEE}
}

@inproceedings{cbf_journal,
  title={Control barrier functions: Theory and applications},
  author={Ames, Aaron D and Coogan, Samuel and Egerstedt, Magnus and Notomista, Gennaro and Sreenath, Koushil and Tabuada, Paulo},
  booktitle={2019 18th European control conference (ECC)},
  pages={3420--3431},
  year={2019},
  organization={Ieee}
}

@inproceedings{zhang2022control,
  title={Control barrier function meets interval analysis: Safety-critical control with measurement and actuation uncertainties},
  author={Zhang, Yuhao and Walters, Sequoyah and Xu, Xiangru},
  booktitle={2022 American Control Conference (ACC)},
  pages={3814--3819},
  year={2022},
  organization={IEEE}
}

@article{lindemann2024learning,
  title={Learning robust output control barrier functions from safe expert demonstrations},
  author={Lindemann, Lars and Robey, Alexander and Jiang, Lejun and Das, Satyajeet and Tu, Stephen and Matni, Nikolai},
  journal={IEEE Open Journal of Control Systems},
  volume={3},
  pages={158--172},
  year={2024},
  publisher={IEEE}
}

@article{cosner2023robust,
  title={Robust safety under stochastic uncertainty with discrete-time control barrier functions},
  author={Cosner, Ryan K and Culbertson, Preston and Taylor, Andrew J and Ames, Aaron D},
  journal={arXiv preprint arXiv:2302.07469},
  year={2023}
}

@inproceedings{ramadan2024control,
  title={A control approach for nonlinear stochastic state uncertain systems with probabilistic safety guarantees},
  author={Ramadan, Mohammad S and Alsuwaidan, Mohammad and Atallah, Ahmed and Herbert, Sylvia},
  booktitle={2024 American Control Conference (ACC)},
  pages={4924--4929},
  year={2024},
  organization={IEEE}
}

@article{alamo2005guaranteed,
  title={Guaranteed state estimation by zonotopes},
  author={Alamo, Teodoro and Bravo, Jos{\'e} Manuel and Camacho, Eduardo F},
  journal={Automatica},
  volume={41},
  number={6},
  pages={1035--1043},
  year={2005},
  publisher={Elsevier}
}

@article{jaulin2002nonlinear,
  title={Nonlinear bounded-error state estimation of continuous-time systems},
  author={Jaulin, Luc},
  journal={Automatica},
  volume={38},
  number={6},
  pages={1079--1082},
  year={2002},
  publisher={Elsevier}
}

@inproceedings{silvestre2024nonlinear,
  title={Nonlinear observers with tighter online error bounds},
  author={Silvestre, Jo{\~a}o Pedro and Nanayakkara, Rahal and Tabuada, Paulo},
  booktitle={2024 IEEE 63rd Conference on Decision and Control (CDC)},
  pages={7728--7733},
  year={2024},
  organization={IEEE}
}

@article{xiao2021high,
  title={High-order control barrier functions},
  author={Xiao, Wei and Belta, Calin},
  journal={IEEE Transactions on Automatic Control},
  volume={67},
  number={7},
  pages={3655--3662},
  year={2021},
  publisher={IEEE}
}

@article{xu2018high,
  title={Constrained control of input--output linearizable systems using control sharing barrier functions},
  author={Xu, Xiangru},
  journal={Automatica},
  volume={87},
  pages={195--201},
  year={2018},
  publisher={Elsevier}
}

@inproceedings{xiao2019high2,
  title={Control barrier functions for systems with high relative degree},
  author={Xiao, Wei and Belta, Calin},
  booktitle={2019 IEEE 58th conference on decision and control (CDC)},
  pages={474--479},
  year={2019},
  organization={IEEE}
}

\clearpage

\section*{APPENDIX}

\subsection{Supporting Lemmas}

\begin{lemma} \label{lem:delta_cbf}
    Suppose that \eqref{eq:cbf_def_2} holds. Then for any compact set $\c X$, there exists $\delta \in \R_{>0}$ such that for all $x \in \c X$:
    \begin{equation*}
        |L_gh(x)| \leq \delta \implies L_fh(x)+\alpha(h(x)) \geq 0.
    \end{equation*}
\end{lemma}
\begin{proof}
    For the sake of contradiction, suppose no such $\delta$ exists. Then, for every $n \in \mathbb Z_{>0}$, there exists some $x_n \in \c X$ such that $|L_gh(x_n)| \leq \frac{1}{n}$, but $L_fh(x_n)+\alpha(h(x_n)) < 0.$
    Since $\c X$ is compact, the sequence $\{ x_n\}$ has a convergent subsequence $x_{n_k} \to x^* \in \c X$. Since $L_gh$ is a continuous function and $\frac{1}{n_k} \to 0$, we have $L_gh(x^*)=0$. Now by continuity of $L_fh$ and $\alpha \circ h$, we also have that \mbox{$L_fh(x^*)+\alpha(h(x^*)) \leq 0.$}
    This contradicts \eqref{eq:cbf_def_2} and thus, such a $\delta>0$ must exist.
\end{proof}

\subsection{Proofs of Main Theorems}

\noindent \textbf{Proof of Theorem \ref{thm:robustly_safe}:}
    For any point $x$ on $\partial \c S$, let $\h x$ be any consistent estimate at any time $t$. Since $k(\h x, t) \in \h K(\h x, t)$, and $\h x$ is a consistent estimate of $x$, we have that 
    \mbox{$a(x)+b(x)k(\h x, t) \geq 0$}. Substituting for $a$ and $b$, we get that $\dot h(x, k(\h x, t)) \geq 0$, and thus $\c S$ is forward invariant \cite{cbf_main}.
\hfill \qedsymbol

\vspace{10pt}
\noindent \textbf{Proof of Theorem \ref{thm:lgh_no_sign_change}:}
    First, note that for any $(\h x, t) \in \R^n \times \R_{\geq 0}$, the set $\t B(\h x,t)$ is compact. 
    Thus, Lemma \ref{lem:delta_cbf} (see Appendix) implies that there exists $\delta_1(\h x, t) \in \R_{>0}$ such that for all $x \in \t B(\h x, t)$:
    \begin{equation*}
        |L_gh(x)| \leq \delta_1(\h x, t) \implies L_fh(x)+\alpha(h(x)) \geq 0.
    \end{equation*}
    Now, consider the case where $L_gh(x) \geq 0$ for all $x \in \c S$. We will show that given any $(\h x, t)$, all $u \in \R$ satisfying:
        \begin{equation*}
        u \geq - \frac{1}{\delta_1(\h x, t)}\min \left\{ 0, \inf_{x \in \t B(\h x, t)} a(x)\right\},
    \end{equation*}
    will be in $\h K(\h x, t)$.
    Note that this choice ensures $u \geq 0$.
    For any $x \in \t B(\h x, t)$, we consider two cases for $b(x) = L_gh(x)$:
    
    \noindent \textbf{Case 1:} $0 \leq b(x) \leq \delta_1(\h x, t)$. By the definition of $\delta_1(\h x, t)$, we have $|b(x)| \leq \delta_1(\h x, t) \implies a(x) \geq 0$. Since we have $u \geq 0$ and $b(x) \geq 0$, it follows that $a(x) + b(x)u \geq 0$.
    
    \noindent \textbf{Case 2:} $b(x) > \delta_1(\h x, t)$. If $a(x) \geq 0$, then $a(x) + b(x)u \geq 0$ trivially holds. If $a(x) < 0$, we have:
    \begin{equation*}
        u \geq - \frac{1}{\delta_1(\h x, t)}\inf_{y \in \t B(\h x, t)} a(y) \geq -\frac{a(x)}{\delta_1(\h x, t)} > -\frac{a(x)}{b(x)}.
    \end{equation*}
    Multiplying both sides by $b(x) > 0$ yields $b(x)u > -a(x)$, which implies $a(x) + b(x)u > 0$.
    
    Thus, for all $x \in \t B(\h x, t)$, the inequality $a(x) + b(x)u \geq 0$ is satisfied, meaning $u \in \h K(\h x, t)$ and the set is non-empty.

    For the alternative case where $L_gh(x) \leq 0$ for all $x \in \c S$, by a similar argument, we can show that any:
    \begin{equation*}
        u \leq \frac{1}{\delta_1(\h x, t)}\min \left\{ 0, \inf_{x \in \t B(\h x, t)} a(x) \right\} \leq 0,
    \end{equation*}
    will be in $\h K(\h x, t)$. 
\hfill \qedsymbol

\vspace{10pt}
\noindent \textbf{Proof of Lemmas \ref{lem:pos_cons} and \ref{lem:neg_cons}:}
    For Lemma \ref{lem:pos_cons}, when $x \in \c C_\gamma$, we have that $h_k(x) \geq 0$ for all \mbox{$k=2,\dots,\gamma$}.
    Now, \eqref{eq:lgh_pos_cons} implies that $L_gh_k(x) \geq \varepsilon_k$ for $k=1,\dots,\gamma-1$. 
    The proof for Lemma \ref{lem:neg_cons} follows similarly.
\hfill \qedsymbol

\vspace{10pt}
\noindent \textbf{Proof of Theorem \ref{thm:odd_gamma}:}
    For the case when $L_g^\gamma h > 0$, we use a positive recursive CBF \eqref{eq:lgh_pos_cons}. 
    By Assumption \ref{assump:nonzero} and \eqref{eq:higher_lie_pos}, we have:
    $$L_gh_\gamma (x) = L_g^\gamma h (x) \geq \varepsilon_\gamma > 0,$$ 
    for all $x \in \c S$, 
    combined with the results of Lemma \ref{lem:pos_cons}, we have that $h_k$ is a CBF for \eqref{eq:si_sys} on $\c C_\gamma$ for all $k=1,\dots,\gamma$. 
    Now, by Corollary \ref{cor:lgh_eps} the $k^\text{th}$ CBF inequality is satisfied by:
    \begin{equation*}
    u \geq - \frac{1}{\varepsilon_k}\min \{ 0, \inf_{x \in B(\h x, t) \cap \c C_\gamma} (L_fh_k(x) + \alpha(h_k(x)))\} =: \underline{u}_k.
    \end{equation*}
    Hence any input $u$ satisfying:
    \begin{equation*}
        u \geq \max_{k=1,\dots,\gamma} \{ \underline{u}_k \},
    \end{equation*}
    satisfies all $\gamma$ CBF inequalities and renders $\c C_\gamma$ robustly safe.

    Alternatively, when $L_g^\gamma h < 0$, we use a negative recursive CBF \eqref{eq:lgh_neg_cons}, and by Assumption \ref{assump:nonzero} and \eqref{eq:higher_lie_neg}, we have that:
    $$L_g \t h_\gamma(x) = (-1)^{\gamma-1}L_g^\gamma h(x) = L_g^\gamma h(x) \leq -\varepsilon_\gamma < 0.$$
    Combined with the results of Lemma \ref{lem:neg_cons}, this implies that $\t h_k$ is a CBF for \eqref{eq:si_sys} on $\t{\c C}_\gamma$ for all $k=1,\dots,\gamma$.
    Now by Corollary \ref{cor:lgh_eps} the $k^\text{th}$ CBF inequality is satisfied by:
    \begin{equation*}
        u \leq \frac{1}{\varepsilon_k}\min \{ 0, \inf_{x \in B(\h x, t) \cap \t{\c C}_\gamma} (L_f\t h_k(x) + \alpha(\t h_k(x)))\} =: \overline{u}_k.
    \end{equation*}
    Hence, any input $u$ satisfying:
    \begin{equation*}
        u \leq \min_{k=1,\dots,\gamma} \{ \overline{u}_k \},
    \end{equation*}
    satisfies all $\gamma$ CBF inequalities and renders the set $\t {\c C}_\gamma$ robustly safe.
\hfill \qedsymbol

\vspace{10pt}
\noindent \textbf{Proof of Theorem \ref{thm:even_gamma}:}
    The proof for rendering $\c C_\gamma$ robustly safe is identical to the first part of the proof of Theorem \ref{thm:odd_gamma}, where a positive recursive CBF \eqref{eq:lgh_pos_cons} is used. To show that $\t{\c C}_\gamma$ can be made robustly safe, we use a negative recursive CBF \eqref{eq:lgh_neg_cons} instead. Now since $L_g^\gamma h(x) \geq \varepsilon_\gamma$, we have that:
    $$L_g \t h_\gamma(x) = (-1)^{\gamma-1}L_g^\gamma h(x) = -L_g^\gamma h(x) \leq -\varepsilon_\gamma < 0.$$
    The rest of the proof follows similarly to the second part of the proof of Theorem \ref{thm:odd_gamma}.
\hfill \qedsymbol

\vspace{10pt}
\noindent \textbf{Proof of Theorem \ref{thm:poly_h_const_g}:}
    Since $h$ is a polynomial of degree $\gamma$, $L_gh(x) = \nabla_x h(x) b$ is a polynomial of degree $\gamma-1$.  Likewise $L^k_gh$ is of degree $\gamma-k$ and thus $L^\gamma_gh$ is of degree 0.
\hfill \qedsymbol

\vspace{10pt}
\noindent \textbf{Proof of Corollary \ref{cor:strong_convexity}:}
    Since $\nabla _x g(x)=0$ and \eqref{eq:strong_convex} holds, we have that $L_g^2h(x) = b^\top \nabla^2_xh(x) b \geq \mu \| b \|^2 > 0$.
\hfill \qedsymbol

\vspace{10pt}
\noindent \textbf{Proof of Corollary \ref{cor:strongly_convex_Lgh2}:}
    Noting that $\nabla_x g(x) = P \nabla^2_x h(x)$ and $P = P^\top$, \eqref{eq:lgh2} directly simplifies to:
    \begin{equation*}
        L_g^2h(x) = 2 \nabla_x h(x) P \nabla^2_x h(x) P \nabla_x h^\top(x).
    \end{equation*}
    Combining \eqref{eq:strong_convex} and the fact that $P$ is symmetric and non-singular, we get
    $P \nabla^2_x h(x) P \succeq \mu \lambda_{\min}(P^2) I$, resulting in:
    \begin{equation*}
        L_g^2h(x) \geq 2 \mu \lambda_{\min}(P^2) \| \nabla_x h(x) \|^2.
    \end{equation*}
     Since $h$ is strongly convex, it holds that $\| \nabla_xh(x) \|^2 \geq \mu^2 \| x-x^* \|^2$ \cite{boyd2004convex}, giving us:
    \begin{equation*}
        L_g^2h(x) \geq 2 \mu^3 \lambda_{\min}(P^2) \| x-x^* \|^2.
    \end{equation*}
    Since $\c S$ is closed and $x^* \notin \c S$, there exists some $\delta \in \R_{>0}$ such that $\| x-x^* \| \geq \delta$ for all $x \in \c S$. Therefore, $L_g^2h(x) \geq 2 \lambda_{\min}(P^2) \delta^2 \mu^3 =: \varepsilon > 0$.
\hfill \qedsymbol

\subsection{Additional Simulations}

In the following example, we illustrate how a system may satisfy Assumption \ref{assump:nonzero} without falling into any of the classes discussed in Sections \ref{sct:polynomial} and \ref{sct:convex}.

\begin{example}
Consider the dynamics of a 2D robotic arm:
\begin{equation}
    \begin{pmatrix}
        \dot x_1 \\ \dot x_2
    \end{pmatrix} = 
    \begin{pmatrix}
        x_2 \\ -b_1 \sin(x_1) - b_2 x_2
    \end{pmatrix}
    + \begin{pmatrix}
        0 \\ 1
    \end{pmatrix} u,
    \label{eq:robotic_arm_sys}
\end{equation}
where the states $x_1$ and $x_2$ represent the position and velocity respectively. The constants $b_1, b_2 \in \R_{>0}$, represent lumped physical quantities.
We define our safe region $\c S$ as the 0-superlevel set of $h(x) = b_1(1-\cos(x_1)) + \frac{1}{2}x_2^2 + 2x_1 x_2 - 1$, which represents the total mechanical energy of the system along with a ``cross term'' $2x_1 x_2$.

\begin{figure}[t]
    \centering
    \includegraphics[width=0.8\linewidth]{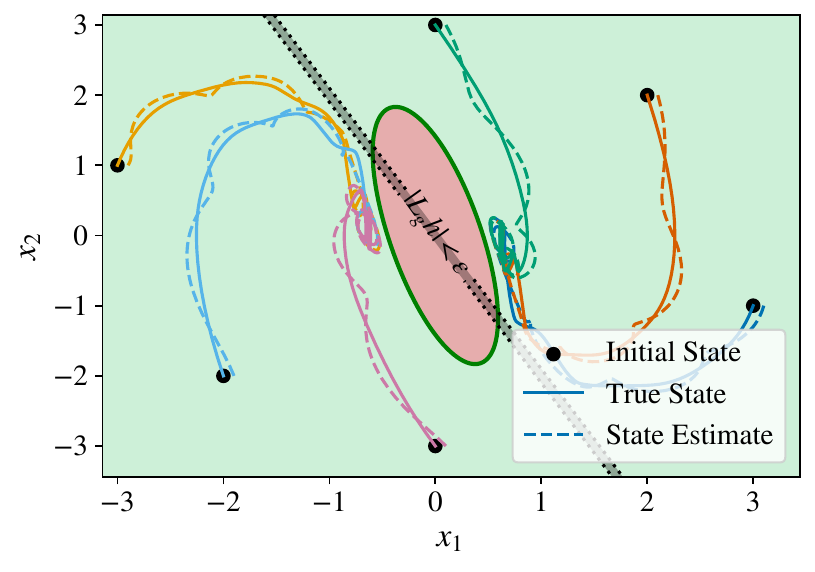}
    \caption{Closed-loop trajectories of \eqref{eq:robotic_arm_sys} with initial conditions marked by black dots. The shaded grey area represents the region where $|L_gh(x)| \leq \varepsilon$.}
    \label{fig:robot_arm_ex}
\end{figure}

We consider the case $b_1=10$ and $b_2=0.1$. For these parameters, $h$ is not even a CBF, since it fails to satisfy \eqref{eq:cbf_def_2}. Thus, the standard CBF condition based on $h$ cannot be used to certify forward invariance of $\c S$, even in the absence of state uncertainty. However, our framework allows us to render a subset of this safe region robustly forward invariant, since $L_gh(x)=2x_1+x_2$ and $L_g^2h=1$, allowing us to apply Theorem \ref{thm:even_gamma}.

Fig.~\ref{fig:robot_arm_ex} depicts the evolution of both the true and estimated trajectories of \eqref{eq:robotic_arm_sys}, where the estimate $\h x=x+e$, with $\|e\|_\infty\leq0.1$, is used to construct a robustly safe input. The value $\varepsilon=0.1$ is used to construct the barriers $h_2$ and $\t h_2$. We observe in Fig.~\ref{fig:robot_arm_ex} that these barriers ensure that the system remains in the region where $|L_gh(x)|\geq\varepsilon$ and a robustly safe control input always exists.

\end{example}

\end{document}